\documentclass[letterpaper, 10 pt, conference]{ieeeconf} 
\IEEEoverridecommandlockouts
\usepackage{cite}
\usepackage[cmex10]{amsmath}
\usepackage{array}
\usepackage{moreverb}
\usepackage{algorithm,algorithmic}
\usepackage{arrayjobx}
\makeatletter
\let\NAT@parse\undefined
\makeatother
\usepackage[colorlinks,citecolor=darkblue,urlcolor=red, linkcolor=black, hyperfigures]{hyperref}
\usepackage{amsmath,amssymb}
\usepackage{times}
\usepackage{graphicx}
\usepackage{subfigure}
\usepackage{setspace}
\usepackage{soul, xcolor}
\usepackage{float}
\usepackage{indentfirst}
\usepackage{bm}
\usepackage{booktabs}
\usepackage{flushend}
\usepackage{balance}
\usepackage[export]{adjustbox}
\usepackage{caption}
\usepackage{enumerate}
 
\usepackage{algorithmic}
 
\makeatletter
\newcommand{\removelatexerror}{\let\@latex@error\@gobble}
\makeatother
\newtheorem{example}{Example}
\newtheorem{theorem}{Theorem}
\newtheorem{lemma}{Lemma}

\newtheorem{corollary}{Corollary}
\newtheorem{definition}{Definition}
\newtheorem{proposition}{Proposition}
\newtheorem{problem}{Problem}
\newtheorem{remark}{Remark}
\newtheorem{claim}{Claim}
\newtheorem{assumption}{Assumption}
\newcommand{\bdefinition}{\begin{definition} \begin{rm} }
\newcommand{\edefinition}{ \end{rm} \hfill \rule{1.5mm}{1.5mm}
\end{definition} }
\newcommand{\bremark}{\begin{remark} \begin{rm} }
\newcommand{\eremark}{ \end{rm} \hfill \rule{1.5mm}{1.5mm}
\end{remark} }
\newcommand{\btheorem}{\begin{theorem}  \begin{rm} }
\newcommand{\etheorem}{ \end{rm} \hfill \rule{1.5mm}{1.5mm}
\end{theorem} }
\newcommand{\blemma}{\begin{lemma} \begin{rm} }
\newcommand{\elemma}{ \end{rm} \hfill \rule{1.5mm}{1.5mm}
\end{lemma} }
\newcommand{\bcorollary}{\begin{corollary} \begin{rm} }
\newcommand{\ecorollary}{ \end{rm} \hfill \rule{1.5mm}{1.5mm}
\end{corollary} }
\newcommand{\bproposition}{\begin{proposition} \begin{rm} }
\newcommand{\eproposition}{ \end{rm} \hfill \rule{1.5mm}{1.5mm}
\end{proposition} }
\newcommand{\bclaim}{\begin{claim} \begin{rm} }
\newcommand{\eclaim}{ \end{rm} \hfill \rule{1.5mm}{1.5mm}
\end{claim} }
\newcommand{\bproblem}{\begin{problem} \begin{rm} }
\newcommand{\eproblem}{\end{rm} \end{problem} }
\newcommand{\bassumption}{\begin{assumption} \begin{rm} }
\newcommand{\eassumption}{ \end{rm} \hfill \rule{1.5mm}{1.5mm}
\end{assumption} }

\definecolor{darkblue}{rgb}{0.0, 0.0, 0.55}

\usepackage[textwidth=1.9cm,color=green!10,textsize=footnotesize]{todonotes}

\def\BibTeX{{\rm B\kern-.05em{\sc i\kern-.025em b}\kern-.08em
    T\kern-.1667em\lower.7ex\hbox{E}\kern-.125emX}}

\title{\Huge \bf
Existential Opacity for Discrete-Event Systems with State Observations
}

\author{Zhiyuan Huang, Zhao Tong, Jiakai Li, and Bingzhuo Zhong\textsuperscript{*}
\thanks{This work was supported by the Guangdong Provincial Project (No. 2024QN11X053) and by the Youth S$\&$T Talent Support Programme of GDSTA (No. SKXRC2025468). ({\sl Corresponding author: Bingzhuo Zhong.})}
\thanks{
Zhiyuan Huang and Bingzhuo Zhong are with the Thrust of Artificial Intelligence, The Hong Kong University of Science and Technology (Guangzhou), Guangzhou 511400, China. (e-mail: zhuang655@connect.hkust-gz.edu.cn, bingzhuoz@hkust-gz.edu.cn).
Zhao Tong and Jiakai Li are with the College of Education Sciences, The Hong Kong University of Science and Technology (Guangzhou), Guangzhou 511400, China. (e-mail: \{ztong837, jli196\}@connect.hkust-gz.edu.cn.)}
}

\begin{document}

\maketitle

\begin{abstract}
Opacity is a fundamental system property for confidentiality in discrete-event systems (DES). 
Classical opacity is typically defined under event-based observations, requiring that any secret system behavior remains indistinguishable from some non-secret behavior to an external intruder.
However, in many applications such as path planning or opacity-preserving tasks, the intruder observes system states rather than events. Moreover, it often suffices that the system exhibits secret behaviors that can be exploited for opacity-preserving task execution, but such a system property cannot be fully captured by existing notions of state-observation-based opacity.
Motivated by this limitation, we propose a relaxed notion of existing state-observation-based opacity, called \emph{existential opacity (EO)},
which only requires the existence of secret behaviors (instead of all secret behaviors) that are indistinguishable from a non-secret behavior under the state observations of the intruder.
We show that the notion of EO is more expressive than existing state-observation-based opacity notions. In addition, a class of EO properties together with their corresponding verification approaches are developed, enabling the analysis of existential opacity in discrete-event systems and providing a new criterion for determining the feasibility of opacity-preserving problems.
\end{abstract}


 \section{Introduction}
Discrete-event systems (DES) are widely used in control theory for systems whose dynamics evolve through event-driven transitions \cite{lafortune2019discrete}. Their well-established mathematical structure enables systematic modeling and analysis of complex systems and supports high-level logical reasoning and verification. As a result, DES have been widely applied in many areas, such as autonomous robotics \cite{zhao2025no}, traffic networks \cite{liang2021application}, and software systems \cite{saadawi2013principles}.
With the increasing complexity of modern applications, besides logical specification satisfaction and physical safety, information-flow security at the system level has also attracted growing research attention \cite{zhong2023towards}. 
In particular, enforcing information confidentiality over DES has become an important and challenging problem.

Opacity has become an important formal notion for confidentiality in the information-flow security of DES \cite{cassandras2007introduction}. In general, opacity concerns the ability of a system to conceal secret behaviors from identification by an external passive intruder that can only see a set of observable events \cite{lin2011opacity}.
Accordingly, a system satisfies opacity if, for any secret behavior, there exists another non-secret behavior that is observationally equivalent from the perspective of the intruder \cite{lin2011opacity, mayer2024current}. 
Here, a behavior of a DES is represented by a sequence of events or states generated by the system. Moreover, a behavior is considered secret if it satisfies a given secret property. Depending on the type of secret behavior that the intruder attempts to infer, different notions of opacity have been proposed. For secrets associated with reaching secret states, classical notions include current-state opacity \cite{saboori2007notions}, initial-state opacity \cite{han2023strong,saboori2008verification}, K-step opacity \cite{saboori2011verificationKstep, yin2017new, balun2022verification}, and infinite-step opacity \cite{saboori2011verification}.
According to different opacity properties, the intruder cannot determine whether the secret behavior has occurred, is occurring, or will occur.
In recent years, opacity notions with language-based secret requirements have also been investigated \cite{wintenberg2022general}, where secrets are defined over system executions rather than visiting secret states. 
Moreover, the concept of opacity has been expanded to cyber-physical systems with the observation of the state output \cite{mayer2024current}.
Nevertheless, the definition of opacity can be restrictive, as it requires all secret behaviors of the system to be indistinguishable to the intruder from some non-secret behavior.

In recent years, the opacity-preserving problem has received increasing attention in cyber-physical systems.
Such a problem mainly focuses on synthesizing a controller such that the secret behavior of the system is indistinguishable from at least one non-secret behavior under the observations of the intruder.
Based on the classical definition of opacity, some existing works \cite{hadjicostis2018trajectory, saboori2011coverage, shi2023security} have investigated the opacity-preserving problem where the intruder infers the system behavior through event observations.
In contrast, in many practical scenarios, the external passive intruder may (partially) observe state-related information of the system instead of events.
Under such state observation, different notions of secrets have been considered in the literature. 
For instance, opacity-preserving problems with secret states have been studied in \cite{yang2020secure, yu2022security}, while execution-level security, where the secret is specified as a property of the entire execution, has been investigated in \cite{wang2020hyperproperties, zheng2023optimal}.
Such problems commonly arise in autonomous robotic systems, especially in path-planning tasks where a robot needs to reach a goal while concealing sensitive information from an external intruder.
However, existing approaches focus solely on generating executions or synthesizing controllers that satisfy specific security requirements without investigating the underlying system property that guarantees the existence of such executions or controllers.

Motivated by the limitation of existing state-observation-based opacity notions, this paper considers a new notion of opacity with respect to state observations by a passive intruder.
The main contributions of this paper are summarized as follows.
\begin{enumerate}
    \item We introduce the notion of \emph{existential opacity (EO)} that captures the existence of secret behaviors that remain indistinguishable from non-secret behaviors to an external intruder that observes the system states. 
    Unlike state-observation-based opacity, which provides only a sufficient condition for the existence of indistinguishable behaviors, EO generalizes this notion by offering a necessary and sufficient condition under which the opacity-preserving problem is feasible.

\item We focus on a subset of EO characterizing a system in which there exist secret behaviors whose state observation sequence is indistinguishable from that of some non-secret behavior to the intruder until the current step.
This type of EO captures the existence of behaviors that prevent the intruder from determining whether the system is currently in secret states.
Focusing on this notion, we develop a verification method for transition systems, which provides a criterion for determining the feasibility of the opacity-preserving problem.

\end{enumerate}
Unlike conventional weak opacity, which is typically formulated over
languages of event sequences and event observations \cite{lin2011opacity}, EO is defined over state sequences and the corresponding state observations.
Compared with the notion of state-observation-based opacity in \cite{mayer2024current}, which requires indistinguishability for all secret behaviors, EO relaxes this requirement by characterizing the existence of at least one indistinguishable secret behavior.
Moreover, EO provides a sufficient and necessary condition for the existence of indistinguishable secret behaviors, allowing one to determine the feasibility of opacity-preserving problem, which cannot be handled by existing results in~\cite{yang2020secure,wang2020hyperproperties, yu2022security, zheng2023optimal}.

The remainder of the paper is organized as follows.
Section \ref{preliminary} presents the preliminaries, including the system model, the intruder model, and the notions of state-observation-based opacity and opacity-preserving problem defined with respect to these models.
Section \ref{sec: motivation} illustrates the limitations of the notion of opacity through a motivating example and introduces the problem considered in this paper.
Section \ref{sec: EO} introduces existential opacity and investigates its relation to the notion of opacity defined in Section \ref{preliminary}.
Section \ref{sec: verfication EO} introduces a specific type of EO and proposes methods for verifying this property.
Section \ref{sec: case study} presents case studies that illustrate the effectiveness of the proposed verification methods.
Finally, Section \ref{Conclusion} concludes the paper.
 \section{Preliminaries} \label{preliminary}
This section reviews the abstractions of the system and intruder models, as well as the notion of opacity defined based on these models. We use $\mathbb{R}$, $\mathbb{R}_{+}$ and $\mathbb{N}$ to denote the set of real numbers, positive real numbers and natural numbers, respectively.
In analogy with the Kleene star and $\omega$-words, the superscripts $*$ and $\omega$
denote finite and infinite paths, respectively, while the superscript $*\omega$ denotes their union.
Let $Proj_{i}$ denote the projection onto the $i$-th component of a tuple.  
For example, for the tuple $(a,b,c,d)$, we have $a = Proj_{1}(a,b,c,d)$.

\vspace{-3pt}
\subsection{Motion ability abstraction}
\begin{definition} \label{def: TS}
Consider a system work in a workspace $\Pi \subseteq \mathbb{R}^{n}$, which is partitioned as $N \in \mathbb{N}$ regions $\Pi =\left \{ \pi_{1}, \dots, \pi_{N}\right \} $.
The motion ability of the system is abstracted as a finite transition system (TS) defined as
\begin{equation} \label{eq: wTS}
    T_{} := (\Pi_{}, \Pi_{0}, U, \to_{}, w, AP, L ),
\end{equation}
where $\Pi_{}$ is a finite set of states representing all regions, 
$\Pi_{0} \subseteq \Pi_{}$ is the set of initial states, $U$ is the set of control inputs,
$\to_{} \subseteq \Pi_{} \times U \times \Pi_{}$ is the transition relation, $w: \Pi \times U \times  \Pi \to \mathbb{R}_{+}$ is a cost function indicating the moving cost between two regions,
$AP$ is a set of atomic propositions having properties of system interests,
and $L: \Pi \to 2^{AP}$ is the labeling function that denotes the properties at a specific region.
The transition system $T$ is assumed to be deterministic, i.e., for any $\pi \in \Pi$ and $u \in U$, there exists at most one $\pi^{+} \in \Pi$ such that $(\pi,u,\pi^{+}) \in \to$.
\end{definition}

An infinite path $\tau =\pi(0)\pi(1) \cdots \in Path^{\omega}(T) $ of $T$ is an infinite sequence of states such that $\pi(0) \in \Pi_{0}$ and $\exists u(k) \in U$ such that $(\pi(k),u(k),\pi(k+1)) \in \to$, for $k\in \mathbb{N}$.
Moreover, the $n$-step finite path for the transition system $T$ is denoted by $\tau(0,n) :=\pi(0)\pi(1)  \cdots \pi(n) \in Path^{*}(T)$ with $n \in \mathbb{N}$.

\subsection{Intruder model and Opacity}
The security property considered in this paper is to prevent the intruder from inferring the identity of the agent that executed certain behaviors of critical importance. To this end, we aim to ensure that such behaviors remain indistinguishable from other non-critical behaviors under the observation of passive intruders. This naturally motivates the adoption of the opacity concept, which provides a formal methodology to model and analyze whether the occurrence of secret behaviors can be concealed from an external passive intruder.

Considering the transition system $T$ as in Definition \ref{def: TS}, the intruder model for the transition system $T$ is defined as follows.
\begin{definition} \label{def: intruder}
Considering a transition system $T = (\Pi, \Pi_{0},U, \to_{}, w, AP, L )$ defined in Definition \ref{def: TS}, 
the intruder model for the transition system $T$ is described by the following observation function 
\begin{equation} \label{intruder_sub}
    H : \Pi \to Y,
\end{equation}
where $Y$ represents the observation output of the intruder based on the state of the transition system $T$.
Furthermore, we assume that the intruder has full knowledge of system model $T$, but can obtain information about the system only through the observation outputs $Y$.
For a path $\tau\in Path^{*\omega}(T)$ generated by the system $T$ starting from an initial state $\pi(0) \in \Pi_{0}$, the corresponding observation sequence by the intruder is defined as $H(\tau) := H(\pi(0)) H(\pi(1)) \cdots\in Path^{*\omega}(Y)$.
\end{definition}
From Definition \ref{def: intruder}, the intruder model we consider in this paper can observe state-related information.
This reflects real-world scenarios, such as remotely operated systems, where state information may be continuously monitored by an intruder through data transmission.
To describe system security in such information flows, we use the notion of opacity.
To this end, we first define the sets of admissible, secret, and non-secret paths as follows, before formally introducing opacity.


\begin{definition} \label{def: path set}
We define the \emph{admissible path set} of system $T$ as $\mathcal{D}(T) \subseteq Path^{*\omega}(T)$, i.e., the set of paths satisfying certain properties such as reaching a goal region or satisfying temporal logic tasks.
The subset $\varphi_s \subseteq \mathcal{D}(T)$ is defined as the \emph{secret path set}, representing the set of secret paths exhibiting secret behaviors, e.g., reaching a secret region or
satisfying a secret temporal logic task.
Finally, the \emph{non-secret path set} $\bar{\varphi}_s$ is defined as the complement of the secret path set, i.e., $\bar{\varphi}_s := \mathcal{D}(T) \setminus \varphi_s$, consisting of all admissible paths that do not involve secret behaviors.
\end{definition}
Intuitively, admissible path set $\mathcal{D}(T)$ specifies the set of system paths under consideration, serving as the domain over which the secret and non-secret path sets are defined. 
Moreover, the secret path set $\varphi_s$ characterizes the class of secret behaviors within the admissible path set under consideration.
Correspondingly, the non-secret path set $\bar{\varphi}_s$ provides the basis for protecting these secret behaviors under intruder observation.

Next, adapted from \cite{mayer2024current}, the state-observation-based opacity property for the transition system $T$ is defined as follows. 

\begin{definition}\label{def: opacity}
    Consider a transition system $T =(\Pi_{}, \Pi_{0}, U, \to_{},w, AP, L ) $ as in Definition \ref{def: TS}, and a passive intruder $H$ for $T$ is modeled in Definition \ref{def: intruder}.
    Moreover, let the admissible path set $\mathcal{D}(T)$, secret path set $\varphi_s$ and non-secret path set $\bar{\varphi}_s$ as defined in Definition \ref{def: path set}.
    The transition system $T$ is $\varphi_{s}-$opaque if $\forall \tau \in \varphi_{s}$, $\exists \tau^{\prime} \in \bar{\varphi}_{s} $, such that $H(\tau) = H(\tau^{\prime})$.
\end{definition}

Based on the notion of opacity in Definition~\ref{def: opacity}, we next define the \emph{opacity-preserving problem}, which aims to design control inputs to ensure that the controlled system remains opaque.

\begin{definition} \label{def: opacity-preserving problem}
Consider a transition system $T=(\Pi, \Pi_{0}, U, \to, w, AP, L )$ as defined in Definition~\ref{def: TS}, and let $H$ be a passive intruder for $T$ as modeled in Definition~\ref{def: intruder}. 
Moreover, we consider the admissible, secret, and non-secret path sets 
$\mathcal{D}(T)$, $\varphi_s$, and $\bar{\varphi}_s$, respectively, 
as defined in Definition~\ref{def: path set}.
Define a control strategy as 
$\mathcal{C} := u(0)u(1)\cdots u(k)\cdots \in Path^{*\omega}(U)$, 
where $Path^{*\omega}(U)$ denotes the set of all control input sequences. 
The corresponding path generated by $\mathcal{C}$ is 
$\tau_{\mathcal C} := \pi(0)\pi(1)\cdots \pi(k)\cdots \in Path^{*\omega}(T)$, 
where $(\pi(k), u(k), \pi(k+1)) \in \to$ for any $k \in \mathbb{N}$.
The \emph{opacity-preserving problem} consists of designing a control strategy 
$\mathcal{C} \in Path^{*\omega}(U)$ such that the path generated by $\mathcal{C}$ satisfies: 
if for any $\tau_{\mathcal C} \in \varphi_s$, then there exists $\tau' \in \bar{\varphi}_s$ such that $H(\tau_{\mathcal C}) = H(\tau')$.
Additionally, the opacity-preserving problem is \emph{feasible} if a control strategy for the opacity-preserving problem exists.

\end{definition}

\begin{remark} \label{remark: opacity}
Note that the admissible path set $\mathcal{D}(T)$ and the secret path set $\varphi_{s}$ jointly determine which type of opacity is evaluated. 
For instance, if $\mathcal{D}(T)$ is the set of all finite paths of $T$ and the secret path set $\varphi_{s}$ specifies being currently in secret states, the resulting notion could reduce to existing notions of opacity such as current-state opacity based on state outputs (CSO-SO) \cite{mayer2024current}.
\end{remark}
\begin{remark}
The transition system $T$ being $\varphi_s$-opaque provides a sufficient condition for the feasibility of the opacity-preserving problem. This follows directly from Definition~\ref{def: opacity}, since for any secret path, there already exists a non-secret path that produces the same observation. 
However, the converse does not necessarily hold. That is, even if the opacity-preserving problem is feasible, the system itself may not be $\varphi_s$-opaque. 
This phenomenon will be illustrated by an example in the next section.
\end{remark}

\section{Motivation and Problem Formulation} \label{sec: motivation}
For a transition system $T$ in Definition \ref{def: TS}, some secret paths $\tau \in \varphi_{s}$ may be inevitably disclosed to the intruder for which no non-secret path $\tau' \in \bar{\varphi}_{s}$ produce the same observation $H(\tau) = H(\tau')$.
This phenomenon is intrinsic to path planning, where physical constraints,
environmental topology, and task requirements may result in certain secret behaviors being uniquely distinguishable by the intruder.
To avoid the disclosure of secret behaviors, the opacity-preserving problem described in Definition \ref{def: opacity-preserving problem} aims to design control inputs such that the system $T$ generates some secret paths that remain indistinguishable from non-secret path to the intruder. 
Accordingly, prior to controller design, it is necessary to verify the system to ensure the feasibility of opacity-preserving problem stated in Definition \ref{def: opacity-preserving problem}. 
However, the notion of state-observation-based opacity in Definition \ref{def: opacity}, which requires secret preservation for all paths in secret path set, is too strict to capture this system property and cannot determine whether the opacity-preserving problem is feasible as in Definition \ref{def: opacity-preserving problem}.
Even in systems classified as non-opaque under this notion, the opacity-preserving problem may still be feasible.

This limitation of the opacity stated in Definition \ref{def: opacity} can be illustrated by means of a simple example. 
\begin{example} \label{eg: motivation}
Consider a workspace $\Pi$ abstracted as Fig. \ref{fig: workspace example},
where the motion of an agent is modeled by a transition system $T$ in Definition \ref{def: TS} whose transitions correspond to the connectivity between states in the workspace abstraction.
Since the system is deterministic, it suffices to consider paths only, as each path corresponds to a unique control strategy. 
The intruder model \eqref{intruder_sub} in this example satisfies $H_{}(q_{1}) =H_{}(q_{3}) =O_{1}$ and $H_{}(q_{2}) =H_{}(q_{4}) =O_{2}$.
Assume that the $q_{2}$ is a secret state and the secret behavior is to visit the state $q_{2}$.
Suppose that the set of initial states for the agent is $\Pi_{0} =\{q_{1}\}$, and that the intruder has complete knowledge of the agent system $T$.
Consider a finite observation sequence by the intruder as $H_{}(\tau(0,2))=O_{1}O_{2}O_{2}$.  
From this observation, the intruder will know for sure that the state $q_{2}$ has been visited by the agent since only the finite path $\tau = q_{1}q_{2}q_{4}$ can render such an observation path. 
In other words, for the finite $\tau = q_{1}q_{2}q_{4}$, there exist no other non-secret finite path $\tau^{\prime} $ satisfying $H(\tau ) = H(\tau^{\prime})$.
This indicates that the WTS of the agent is not opaque according to Definition \ref{def: opacity}.
However, if the finite path of the agent is $\tau(0,2)=q_{1}q_{1}q_{2}$, the sub-intruder cannot infer whether the agent has visited the secret state $q_{2}$, since there exists a non-secret path $\tau^{\prime}(0,2)=q_{1}q_{3}q_{4}$ that produces the same observation $O_{1}O_{1}O_{2}$.
\end{example}
\begin{figure}[htbp]
	\centering
	\includegraphics[width=0.3\textwidth]{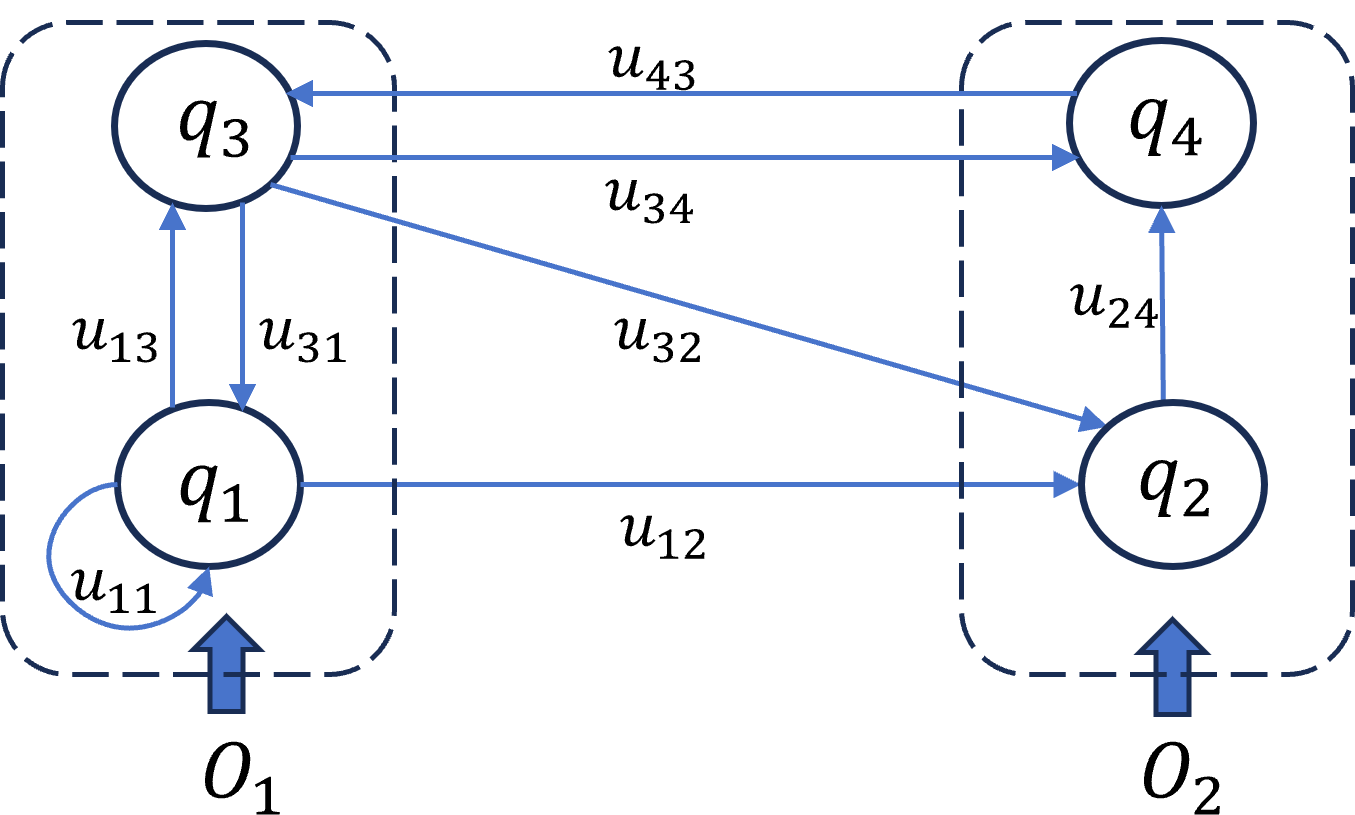}
    \vspace{-4pt}
	\caption{The abstracted workspace.}
    \vspace{-5pt}
	\label{fig: workspace example}
\end{figure}
The example shows that even when the state-observation-based opacity fails, the system still has the secret paths which are indistinguishable from some non-secret path to the intruder.
Thus, state-observation-based opacity in Definition \ref{def: opacity} cannot fully determine whether the opacity-preserving problem is feasible.
Motivated by the limitation of the state-observation-based opacity, we propose, in the next section, a more relaxed notion of opacity, called \emph{existential opacity}, for discrete event systems. 
The EO depicts the system property of existence of secret paths that remain indistinguishable from non-secret paths to the intruder, thereby allowing us to access whether the opacity-preserving problem is feasible.

\section{Existential opacity} \label{sec: EO}

\subsection{Existential Opacity}
To overcome the limitation of opacity for capturing secret behaviors of transition systems $T$ pointed out in Example \ref{eg: motivation}, we introduce the notion of existential opacity.
\begin{definition} \label{def: EO}
    Consider a transition system $T$ as in Definition \ref{def: TS} and a passive intruder $H$ for $T$ as in Definition \ref{def: intruder}. 
    Let the admissible path set, secret path set and non-secret path set be denoted by $\mathcal{D}(T)$, $\varphi_s$, and $\bar{\varphi}_s$, respectively, according to Definition~\ref{def: path set}.
    The transition system $T$ is $\varphi_{s}-$existentially opaque (EO) if  $\exists \tau \in \varphi_{s}$, $\exists \tau^{\prime} \in \bar{\varphi}_{s} $, such that $H(\tau) = H(\tau^{\prime})$.
\end{definition}

Next, we present a theorem showing that the existential opacity defined in Definition~\ref{def: EO} is more expressive than the original notion of opacity in Definition~\ref{def: opacity}.
\begin{theorem}
 Consider a transition system $T$ in Definition \ref{def: TS} and a passive intruder $H$ for $T$ defined in Definition \ref{def: intruder}.
 Furthermore, let $\mathcal{D}(T)$, $\varphi_s$, and $\bar{\varphi}_s$ denote the admissible path set, secret path set and non-secret path set, respectively, following Definition~\ref{def: path set}. 
 If the system $T$ is $\varphi_{s}-$opaque, then it is also $\varphi_{s}-$existentially opaque.
\end{theorem}
The above implication follows directly from Definitions \ref{def: opacity} and \ref{def: EO}.
Specifically, $\varphi_s$-opacity requires that for every secret path there exists
an indistinguishable non-secret path, whereas $\varphi_s$-existential opacity only
requires the existence of such a path. Hence, the former implies the latter.
Note that the converse does not generally hold, as existential opacity only requires the existence of a non-secret path that is observationally indistinguishable from at least one secret path, rather than from all secret paths.
Example \ref{eg: motivation} illustrates a case in which $T$ is $\varphi_{s}-$existentially opaque but not $\varphi_{s}-$opaque.

Next, we investigate the relation between EO in Definition \ref{def: EO} and the feasibility of the opacity-preserving problem in Definition \ref{def: opacity-preserving problem} by the following corollary.
\begin{corollary}\label{corollary: EO feasibility}
Let $T$ be a transition system given in Definition \ref{def: TS}, and $H$ be a a passive intruder for $T$ introduced in Definition \ref{def: intruder}.
The admissible path set, secret path set, and non-secret path set are denoted by $\mathcal{D}(T)$, $\varphi_s$, and $\bar{\varphi}_s$ respectively, following Definition~\ref{def: path set}.
The system $T$ is EO if and only if the opacity preserving problem is feasible in the sense of Definition~\ref{def: opacity-preserving problem}.
\end{corollary}

\begin{proof}
($\Rightarrow$) Suppose that the system $T$ is EO. 
By Definition~\ref{def: EO}, there exists a secret path $\tau \in \varphi_s$ such that there exists a non-secret path $\tau' \in \bar{\varphi}_s$ satisfying
$H(\tau) = H(\tau')$.
Since $T$ is deterministic from Definition \ref{def: TS}, there exists a control strategy $\mathcal{C} \in Path^{*\omega}(U)$ that generates $\tau$, i.e., $\tau = \tau_{\mathcal C}$. 
Thus, the control strategy $\mathcal C$ satisfies the opacity-preserving condition, which implies that the opacity-preserving problem is feasible according to Definition~\ref{def: opacity-preserving problem}.

($\Leftarrow$) Conversely, suppose that the opacity-preserving problem is feasible. 
Then there exists a control strategy $\mathcal C \in Path^{*\omega}(U)$ such that the generated path $\tau_{\mathcal C}$ satisfies
\[
\tau_{\mathcal C} \in \varphi_s \Rightarrow \exists \tau' \in \bar{\varphi}_s \text{ such that } H(\tau_{\mathcal C}) = H(\tau').
\]
In particular, based on Definition \ref{def: opacity-preserving problem}, if $\tau_{\mathcal C} \in \varphi_s$, then there exists a non-secret path $\tau'\in\bar{\varphi}_{s}$ that is indistinguishable from $\tau_{\mathcal C}$ under the observation of the intruder, i.e. $H(\tau) = H(\tau^{\prime})$. 
Hence, the system $T$ satisfies EO based on Definition \ref{def: EO}.
\end{proof}
Corollary \ref{corollary: EO feasibility} shows that EO provides both a necessary and sufficient condition for the feasibility of the opacity-preserving problem stated in Definition \ref{def: opacity-preserving problem}. 
In particular, this property cannot be guaranteed by the notion of state-observation-based opacity as defined in Definition~\ref{def: opacity}, 
as already illustrated in Example~\ref{eg: motivation}.

\begin{remark}
Similar to the discussion in Remark \ref{remark: opacity}, 
by appropriately specifying the admissible path set $\mathcal{D}(T)$ together with the secret path set $\varphi_s$ introduced in Definition \ref{def: path set}, EO can capture different system security properties in opacity-preserving problems.
For instance, let $\mathcal{D}(T)$ consist of the set of all paths of $T$ that satisfy a given task specification. 
If $\varphi_s$ characterizes a set of paths that visit secret states, originate from secret initial states, or satisfy certain secret specifications, 
then EO could capture the corresponding system security properties in \cite{yu2022security}, \cite{yang2020secure}, and \cite{zheng2023optimal}, respectively.
\end{remark}

\section{Verification of Existential Opacity} \label{sec: verfication EO}
In this section, we focus on the verification problem of the proposed EO in Definition \ref{def: EO}. 
As pointed out in Corollary \ref{corollary: EO feasibility}, verifying that a system is EO directly determines the feasibility of the opacity-preserving problem in Definition \ref{def: opacity-preserving problem}, with respect to the corresponding sets of admissible paths and secret paths introduced in Definition \ref{def: opacity}.
Concretely, we focus on path-based secrets, where a secret is incurred when a path visits a designated secret state.
The secret states are modeled as $\pi_{s} \in \Pi_{s} \subset \Pi$. 
Moreover, we define the set of non-secret states as $\Pi_{ns}= \Pi \setminus \Pi_{s}$.

\subsection{ Current-State Existential Opacity} \label{subsec:ISEO CSEO}
We denote by $\tau(0,n): =\pi(0)\pi(1) \cdots \pi(n) \in Path ^{*}(T) $ an $n$-step path generated by the transition system $T$ in Definition \ref{def: TS} and by $H(0,n):= H(\pi(0)) H(\pi(1)) \cdots H(\pi(n))  \in Path^{*}(Y)$ its observation sequence perceived by the intruder.
Next, motivated by the concept of current-state-based opacity \cite{mayer2024current}, we define the following type of existential opacity.
\begin{definition} \label{def: CSEO}
   Consider a transition system $T$ as in Definition~\ref{def: TS}, a passive intruder $H$ defined in Definition~\ref{def: intruder}, and a set of secret state $\Pi_{s} \subset\Pi$. The system $T_g$ is said to be \emph{current-state existentially opaque (CSEO)} if $\exists \tau \in Path^{*\omega}(T)$ such that
   for any step $k \in \mathbb{N}$ along the path $\tau$, if $\pi(k) \in \Pi_s$, then $\exists \tau' \in Path^{*\omega}(T)$ satisfying
        (1)  $\pi'(k) \notin \Pi_s$;
        (2) $H(\tau(0,k)) = H(\tau'(0,k))$.

\end{definition}
From Definition~\ref{def: CSEO}, CSEO characterizes the system property that there exists a system path such that, 
based on the observation sequence up to the current step $n$, the intruder can never be certain that the system is currently in a secret state.
In other words, the secret inference by the intruder is based on the observation sequence of path prefixes up to the current step, i.e., $H(\tau(0,n))$.
Therefore, the admissible path set associated with CSEO is 
\begin{equation} \label{eq: D_CSO}
    D_{CSO}(T) = \{\tau(0,n) \mid \tau \in Path^{* \omega}(T), n \in \mathbb{N} \}.
\end{equation}
The corresponding secret path set is
\begin{equation} \label{eq: phi_s_i}
    \varphi_{s} = \{  \tau(0,n) \mid \pi(n) \in \Pi_{s}, \tau \in Path^{* \omega}(T), n \in \mathbb{N} \}.
\end{equation}
As a result, one can verify that if a transition system $T$ is CSEO, then $T$ is EO with $\mathcal{D}(T)=\mathcal{D}_{CSO}(T)$ and $\varphi_s$ defined as in \eqref{eq: D_CSO} and \eqref{eq: phi_s_i}.

\subsection{Verification Methods}
In this subsection, we develop a verification method for the CSEO defined in Definition~\ref{def: CSEO}.
First, we introduce an observer transition system (OTS) for the transition system $T$ to track the information observed by the intruder $H$.
\begin{definition} \label{def: OTS}
    Given a transition system $T$ as in Definition \ref{def: TS}, its corresponding OTS $T^{o}$ is defined by:
\begin{equation*}
    T^{o} = (\Pi^{o}, \Pi_{0}^{o}, \to^{o}, U),
\end{equation*}
where 
\begin{itemize}
    \item {$\Pi^{o} \subseteq \Pi \times 2^{\Pi}\times M$ is the set of observer states, where each observer state $\pi_{}^{o} \in \Pi^{o}$ is a triple $\pi^{o}:=(\pi, C(\pi), m)$.
    Here, 
    \begin{align}
        C: \Pi \to 2^{\Pi} \label{c-set}
    \end{align}
    is a consistent-state mapping induced by the transition relation of $T$ and the observation mapping $H$, and $M:= \{\emptyset,\{1\},\{2\}\}$ is a set of markers determined by the first two components $(\pi, C(\pi))$ of the observer state.
    Specifically, for $(\pi,u,\pi^{+ }) \in \to$, the evolution of the consistent-state set is characterized by an update operator $\mathcal{U}$, i.e., $C(\pi^{+}) = \mathcal{U}\big(C(\pi), \pi^{+ }\big)$, where 
    \begin{align} \label{eq: update operator}
        &\mathcal{U}\big(C(\pi), \pi^{+}\big)=   \{\hat{\pi}^{+} \mid \exists \hat{\pi} \in C(\pi),\exists \hat{u} \in U~ \text{such that}~ \nonumber \\& ~~~(\hat{\pi}, \hat{u}, \hat{\pi}^{+})\in \to  \wedge H(\hat{\pi}^{+}) = H(\pi^{+}) \}.
   \end{align} 
    Furthermore, the marker $m \in M$ is determined by the pair $(\pi, C(\pi))$:
    
    (1) $m = \{1\}$, if the first component of the observer state $\pi_{}^{o} \in \Pi^{o}$ satisfies $\pi \in \Pi_{s} $ and there exists a non-secret state $\pi_{ns} \in \Pi_{ns}$ such that  $\pi_{ns} \in C(\pi) \setminus \{\pi\}$;

    (2) $m = \{2\}$, if the first component in observer state $\pi_{}^{o} \in \Pi^{o}$ satisfies $\pi \in \Pi_{s} $ and $C(\pi) \subseteq \Pi_{s}$;


    (3) $m =  \emptyset$, otherwise.
    }
    
    \item {$\Pi_{0}^{o} =\Pi_{0} \times 2^{\Pi_{0}} \times M$ is defined as the initial set for observer states, where each initial observer state $\pi^{o}_{0} \in \Pi_{0}^{o}$ is a triple $\pi^{o}_{0}:=(\pi_{0},C(\pi_{0}), m) $.
    Here, $\pi_{0} \in \Pi_{0}$ denotes the initial state of the transition system $T$ as in Definition \ref{def: TS}. Moreover, $C(\pi_{0})$ satisfies $C(\pi_{0}) = \{ \pi \in \Pi | \pi \in \Pi_{0}, H(\pi) = H(\pi_{0}) \}$, which represents the set of all initial states that yield the same observation to the intruder $H$ for the transition system $T$.
     Moreover, the marker $m \in M$ associated with the initial observer state depends on the pair $(\pi_{0}, C(\pi_{0}))$.

    
    (1) $m = \{1\}$, if the first component of the initial observer state $\pi_{0}^{o} \in \Pi^{o}$ satisfies $\pi_{0} \in \Pi_{s} $ and there exists a non-secret state $\pi_{ns} \in \Pi_{ns}$ such that  $\pi_{ns} \in C(\pi_{0}) \setminus \{\pi_{0}\}$;

    (2) $m = \{2\}$, if the first component in the initial observer state $\pi_{0}^{o} \in \Pi^{o}$ satisfies $\pi_{0} \in \Pi_{s} $ and $C(\pi_{0}) \subseteq \Pi_{s}$;


    (3) $m =  \emptyset$, otherwise.
    }
    \item {$\to^{o} \subseteq \Pi^{o} \times U \times \Pi^{o}$ is the transition relation of the OTS. 
    For any $\pi^{o}_{} = (\pi, C(\pi), m)$, $\pi^{o+} = (\pi^{+}, C(\pi^{+}), m^{+}) \in \Pi^{o}$, we have $(\pi^{o}_{}, u, \pi^{ o+}_{}) \in \to^{o}$ if the following holds:
    
    (1) $(\pi,u,\pi^{+}) \in \to$ and $u \in U$; 
    (2) $C(\pi^{+}) =\mathcal{U}\big(C(\pi), \pi^{+}\big)$;
    (3) $C(\pi^{+})  \ne \emptyset$
    
   }
\end{itemize}
\end{definition}
The OTS in Definition~\ref{def: OTS} integrates the motion capabilities of the transition system $T$ with the observation information obtained from its associated intruder.
Let $Path^{* \omega}(T_i^o)$ denote the set of all trajectories (paths) of the OTS $T^o$.
A path
$\tau^o:=\pi^{o}(0)\pi^{o}(1)\cdots\pi^{o}(k) \cdots  \in Path^{* \omega}(T^o)$
is a sequence of observer states of $T^o$ such that
$(\pi^{o}(k),u(k),\pi^{o}(k+1)) \in \to^o$
for all $k \in \mathbb{N}$, where $\pi^{o}_{}(k) := (\pi(k), C(\pi(k)), m(k))$ denotes the $k$-th step of $\tau^o\in Path^{* \omega}(T^o)$ and $u(k)$ satisfies $(\pi(k),u(k),\pi(k+1)) \in \to$.

Next, the following theorem characterizes the property of the consistent-state set $C(\pi(k))$ as defined in \eqref{c-set} along a path $\tau^o \in Path^{*\omega}(T^o)$, which will be used to verify CSEO.

\begin{theorem} \label{theorem: consistant set}
Consider a transition system $T$ in Definition \ref{def: TS} and its corresponding OTS $T^{o}$ constructed as in Definition \ref{def: OTS}.
For a path $\tau^o=\pi^{o}(0)\pi^{o}(1)\cdots\pi^{o}(k) \cdots  \in Path^{* \omega}(T^o)$ with  $\pi^{o}_{}(k) = (\pi(k), C(\pi(k)), m(k))$, the consistent-state set $C(\pi(k))$ is exactly the set of all reachable states observationally consistent with the execution
prefix $H(\tau(0,k))$ for each step $k$, i.e. 
\begin{align}
  &  C(\pi(k))
= \{\pi'(k) \mid
\exists \tau'(0,k)
\text{ such that }\nonumber \\
&~~~~~~~~~~
H(\tau'(0,k))=H(\tau(0,k))\},
\end{align}
where $\tau'(0,k): = \pi'(0)\pi'(1)\cdots\pi'(k) \in Path^{*}(T)$ and $\pi'(k)$ is the last state in the path $\tau'(0,k)$.
\end{theorem}
\begin{proof}
    We prove the theorem by mathematical induction.

    Base Case: 
From the construction of the initial consistent-state set, we have
$C(\pi(0))
=\{\pi \in \Pi_0 \mid H(\pi) = H(\pi(0))\}$, which includes all initial states that are observationally equivalent to the initial observation $H(\tau(0,0))$. 
Hence the claim holds for $k=0$.

Induction Hypothesis:
Assume that for some $k \geq 0$, the consistent-state set satisfies
$C(\pi(k))=\{\pi'(k) \mid \exists \tau'(0,k)\text{ such that }H(\tau'(0,k)) = H(\tau(0,k))\}$. That is, $C(\pi(k))$ is the set of all states that are observationally
consistent with the execution prefix $H(\tau(0,k))$.

Induction Step.
Consider step $k+1$.
By the definition of the existential update operator as in \eqref{eq: update operator}, we have
\begin{align}
&    C(\pi(k+1))
=\mathcal{U}\big(C(\pi(k)),\pi(k+1)\big)= \nonumber \\ &
\{
\hat{\pi}^+
|
\exists \hat{\pi}\in C(\pi(k)),\nonumber\exists \hat{u} \in U,\text{such that}\\ &
~~~~~~(\hat{\pi}, \hat{u},\hat{\pi}^+)\in\to, 
H(\hat{\pi}^+)=H(\pi(k+1))
\}.
\end{align}

By the induction hypothesis, $C(\pi(k))$ is the observational equivalence
maximal state set corresponding to the path prefix $H(\tau(0,k))$.
Therefore, the update operator $\mathcal{U}$ generates exactly all
reachable successor states that preserve observational consistency with the
extended observation sequence $H(\tau(0,k+1))$.
Hence,
\begin{align}
& C(\pi(k+1))
=
\{\pi'(k+1) \mid
\exists \tau'(0,k+1)
\text{ such that } \nonumber \\&~~~
~H(\tau'(0,k+1))=H(\tau(0,k+1))\}.
\end{align}

Thus the induction holds.
\end{proof}

Intuitively, Theorem \ref{theorem: consistant set} states that after observing the prefix $H(\tau(0,k))$, the intruder cannot distinguish the true state $\pi(k)$ from any state in $C(\pi(k))$.  
Thus, $C(\pi(k))$ represents all states observationally consistent with the observation sequence up to the step $k$.
Next, we derive a corollary of Theorem \ref{theorem: verify CSEO} that associates each state in the consistent-state set with a corresponding path in the transition system $T$, providing the theoretical foundation for the subsequent CSEO verification approach.
\begin{corollary}\label{corollary: Ci_complete}
Given a path $\tau= \pi(0)\pi(1) \dots\pi(k) \cdots \in Path^{*\omega}(T)$ for a transition system $T$ as in Definition \ref{def: TS}, and its corresponding OTS $T^{o}$ is constructed following the Definition \ref{def: OTS}, where the update operator is defined in \eqref{eq: update operator}.
The consistent-state set $C(\pi(k))$ satisfies $ \pi_{c} \in C(\pi(k)) $
if and only if $\exists \tau^{\prime}=\pi^{\prime}(0)\pi^{\prime}(1) \dots\pi^{\prime}(k) \cdots \in   Path^{* \omega}(T)\text{ such that }
H(\tau^{\prime}(0,k)) = H(\tau(0,k))$ and $\pi_{c} =\pi^{\prime}(k)$.
\end{corollary}

\begin{proof}
The Corollary directly follows from the construction of the consistent-state
set and the existential update operator as in \eqref{eq: update operator}.
By Theorem \ref{theorem: consistant set},
$C(\pi(k)) =\{\pi'(k) \mid
\exists \tau'(0,k)
\text{ such that }
H(\tau'(0,k))=H(\tau(0,k))\} $.
Therefore, for any $\pi_c \in C(\pi(k))$, there exists a path $\tau'=\pi'(0)\pi'(1)\dots\pi'(k)\dots \in Path^{*\omega}(T)$
such that $H(\tau'(0,k))=H(\tau(0,k))$
and $\pi_c=\pi'(k)$.

Conversely, suppose that 
$\exists \tau'=\pi'(0)\pi'(1)\dots\pi'(k)\dots \in Path^{*\omega}(T)$ 
such that 
$H(\tau'(0,k))=H(\tau(0,k))$ 
and $\pi_c=\pi'(k)$.
This implies that $\pi_c$ is a state reachable at step $k$ through a path whose observation prefix is identical to that of $\tau(0,k)$.
According to Theorem \ref{theorem: consistant set}, $\pi_{c} \in C(\pi(k))$ holds.

Hence, the corollary holds.
\end{proof}

Intuitively, Corollary \ref{corollary: Ci_complete} states that each state in the consistent-state set $C(\pi(k))$ corresponds to at least one path in the original transition system $T$ that produces the same observation prefix $H(\tau(0,k))$. 
For $\pi^{o}_{} = (\pi, C(\pi), m)$, we denote $\pi^{o} \in \Pi_{s}$ if $\pi \in \Pi_{s}$.
Next, we propose methods to verify the CSEO defined in Definition~\ref{def: CSEO} based on the OTS stated in Definition \ref{def: OTS}.
\begin{theorem} \label{theorem: verify CSEO}
    Consider a transition system $T$ described in Definition \ref{def: TS}, a passive intruder $H$ defined in Definition~\ref{def: intruder}, and a set of secret state $\Pi_{s} \subset\Pi$. 
    The corresponding OTS associated with $T$ is constructed following the Definition \ref{def: OTS}, where the update operator is defined in \eqref{eq: update operator}. 
    The system $T$ is CSEO iff $ \exists\tau^o \in Path^{*\omega}(T^o)$ such that
    $\forall k \in \mathbb{N}, \text{if}~ \pi^{o}(k) \in \Pi_{s}, \text{then~}
    m(k) = \{1\}$ holds.
\end{theorem}
\begin{proof}
    $(\Rightarrow)$ Suppose $T$ is CSEO. 
    Then, according to Definition \ref{def: CSEO}, $\exists\tau \in Path^{*\omega}(T)$ such that 
    for any $k \in \mathbb{N}$, if $\pi(k) \in \Pi_s$, we have $\exists \tau' \in Path^{*\omega}(T)$ satisfying  $\pi'(k) \notin \Pi_s$ and $H(\tau(0,k)) = H(\tau'(0,k))$.
   By the OTS construction in Definition \ref{def: OTS} and Theorem \ref{theorem: consistant set}, the set $C(\pi(k)) = \mathcal{U}
\big(
C(\pi(k-1)),\pi(k)
\big)$ in $\pi^{o}(k)$ characterizes the set consisting of all states in the $k$-th step observationally consistent with the observation sequence $H(\tau(0,k))$.
If the path $\tau:= \pi(0)\pi(1) \cdots \pi(k) \cdots  \in Path ^{*\omega}(T)$ visits secret states and
$T$ is CSEO, then for any step $k_{s}$ such that $\pi(k_{s}) \in \Pi_s$, there exist another path $\tau^{\prime}:=\pi^{\prime}(0)\pi^{\prime}(1) \dots\pi^{\prime}(k_{s}) \cdots \in   Path^{* \omega}(T) $ such that $\pi'(k_{s}) \notin \Pi_s$ and $H(\tau(0,k_{s})) = H(\tau'(0,k_{s}))$. 
Therefore, in the corresponding OTS $T^o$, there exists a path $\tau^o=\pi^{o}(0)\pi^{o}(1)\cdots\pi^{o}(k_{s}) \cdots  \in Path^{* \omega}(T^o)$ such that $\pi^{o}(k_{s})=(\pi(k_{s}), C(\pi(k_{s})), m(k_{s}))$ with $\pi^{\prime}(k_{s}) \in C(\pi(k_{s}))$ and $m(k_{s})=\{1\}$ according to the OTS construction in Definition \ref{def: OTS}.

($\Leftarrow$)
Conversely, suppose there exists a path 
$\tau^o \in Path^{*\omega}(T^o)$ satisfying that 
$ \forall k \in \mathbb{N}$, if $\pi^{o}(k) \in \Pi_{s} $, $m(k) = \{1\}$ holds.
Let $\tau = \pi(0)\pi(1) \cdots \pi(k) \cdots  \in Path ^{*\omega}(T)$ be the corresponding path such that for each step $k$, $\pi(k) = Proj_{1} \pi^{o}(k)$.
Suppose there exists a step $k_{s}$ in $\tau^{o}$ such that $\pi^{o}(k_{s}) \in \Pi_{s} $.
Since $m(k_{s})=\{1\}$, the set $C(\pi(k_{s}))$
contains at least one non-secret state $\pi_{ns}\in C(\pi(k_{s}))$.
By Corollary \ref{corollary: Ci_complete}, there exists a path $\tau^{\prime}:=\pi^{\prime}(0)\pi^{\prime}(1) \dots\pi^{\prime}(k_{s}) \cdots \in   Path^{* \omega}(T) $ such that
$H(\tau(0,k_{s})) = H(\tau'(0,k_{s}))$
and $\pi'(k_{s})\notin\Pi_s$, which satisfies Definition \ref{def: CSEO}.
Therefore, the system $T$ is CSEO.
\end{proof}

Theorem \ref{theorem: verify CSEO} provides a system verification approach for CSEO in terms of the existence of a path in the observation transition system $T^{o}$. However, directly checking the condition in Theorem \ref{theorem: verify CSEO} by enumerating paths of $T^{o}$ is generally impractical, especially when the system has a large or infinite number of paths. 
Therefore, instead of examining paths individually, we transform the verification problem to a reachability check on a modified OTS.
Specifically, since states with marker component $\{2\}$ violate the requirement in Theorem \ref{theorem: verify CSEO}, we remove such states in OTS and construct a reduced observer transition system (ROTS).


\begin{definition} \label{def: sub OTS}
Let $T^{o}_{\neg 2}= (\Pi^{o}_{\neg 2}, \Pi^{o}_{0,\neg 2}, \to^{o}_{\neg 2}, U)$ denote the reduced observer transition system (ROTS) of $T^{o}$ obtained by removing all states $\pi^{o}$ such that $m=\{2\}$ and the transitions incident to them, where
\begin{itemize}
    \item $\Pi^{o}_{\neg 2} = \{\pi^{o} \in \Pi^{o} \mid m \neq \{2\}\}$ contains all states of the original system $T^{o}$ except those whose marker component is $\{2\}$.
    \item $\Pi^{o}_{0,\neg 2} = \Pi^{o}_{0} \cap \Pi^{o}_{\neg 2}$ defines the initial states of the reduced system and consists of the initial states of $T^{o}$ that remain after removing states with marker $\{2\}$.
    \item $\to^{o}_{\neg 2} := \{(\pi^{o},u,\pi^{o+}) \in \to^{o} \mid
\pi^{o}\in \Pi^{o}_{\neg 2},\ \pi^{o+}\in \Pi^{o}_{\neg 2}, u\in U\}$, which indicates that the transition relation $\Pi^{o}_{0,\neg 2} = \Pi^{o}_{0} \cap \Pi^{o}_{\neg 2}$ preserves only those transitions in $\to^{o}$ whose source and target states both belong to $\Pi^{o}_{\neg 2}$.
\end{itemize}
\end{definition}
Based on the ROTS derived from Definition \ref{def: sub OTS}, we present the following corollary that verifies the CSEO property in terms of reachability analysis in ROTS.
\begin{corollary} \label{corollary: ROTS}
Consider a transition system $T$ with a set of secret state $\Pi_{s} \subset\Pi$, a passive intruder $H$ defined in Definition~\ref{def: intruder}, and the corresponding OTS $T^{o}$.
Let $T^{o}_{\neg 2}$ denote the ROTS obtained by removing
all observer states with marker $m=\{2\}$ as defined in Definition \ref{def: sub OTS}.
The system $T$ is CSEO if and only if 
$T^{o}_{\neg 2}$ contains a state $\pi^{o} \in \Pi_s$ with marker $m=\{1\}$ that is reachable from some initial state in $\Pi^{o}_{0,\neg 2}$.
\end{corollary}
\begin{proof}
($\Rightarrow$)
Suppose that $T$ is CSEO. 
By Theorem \ref{theorem: verify CSEO}, there exists a path $\tau^o \in Path^{*\omega}(T^o)$
and an index $k_{s} \in \mathbb{N}$ such that $\pi^o(k_{s} ) \in \Pi_s$ and
$m(k_{s} )=\{1\}$. Moreover, by the CSEO condition, whenever
$\pi^o(k_{s} ) \in \Pi_s$ along the path, it holds that $m(k_{s} )=\{1\}$.
In particular, the path leading to $\pi^o(k_{s})$ never visits a state
with marker $m=\{2\}$.
Therefore, this path remains entirely within the ROTS $T^o_{\neg 2}$, which implies that $\pi^o(k_{s} )$ is reachable in
$T^o_{\neg 2}$. Hence, there exists a reachable state
$\pi^o \in \Pi_s$ with $m=\{1\}$ in $T^o_{\neg 2}$.

($\Leftarrow$)
Conversely, suppose that there exists a reachable state
$\pi^o \in \Pi_s$ with $m=\{1\}$ in $T^o_{\neg 2}$.
By definition of $T^o_{\neg 2}$, all states with marker $m=\{2\}$ have
been removed. Hence, any path in $T^o_{\neg 2}$ corresponds to a path
in $T^o$ that never visits a state with marker $\{2\}$.
Since $\pi^o$ is reachable in $T^o_{\neg 2}$, there exists a path
$\tau^o \in Path^{*\omega}(T^o)$ that reaches $\pi^o$ without passing
through any state in which $m=\{2\}$. Consequently, whenever a state
$\pi^o(k) \in \Pi_s$ occurs along this path, its marker must satisfy
$m(k)=\{1\}$.
Therefore, the condition of Theorem \ref{theorem: verify CSEO} holds, and $T$ is CSEO.
\end{proof}
Intuitively, Corollary \ref{corollary: ROTS} indicates that the verification of CSEO can be addressed by examining the reachable states in the ROTS. 
If there exists a reachable secret state in ROTS with marker $m=\{1\}$, it indicates that the corresponding secret path is observationally indistinguishable from some non-secret path to the intruder up to the current step. 
Therefore, the existence of such a state guarantees that the system satisfies CSEO.
\begin{remark}
The verification of CSEO reduces to a reachability problem in the reduced observation transition system $T^o_{\neg 2}$. 
Such reachability can be checked using standard graph search algorithms, such as breadth-first search or depth-first search, starting from the initial state set $\Pi^o_{0,\neg2}$.
\end{remark}
\vspace{-4pt}
\section{Case Study} \label{sec: case study}
In this section, we consider the same workspace shown in Fig.~\ref{fig: workspace example} and the same intruder model described in Example~\ref{eg: motivation}. 
We assume that the intruder has complete knowledge of the system model of an agent. 
Two cases with different initial and secret state sets are considered. 
Since these sets differ in each case, the corresponding transition system $T$ defined in Definition~\ref{def: TS} are also different.

\emph{Case I:} Suppose that the set of initial states of the agent is $\Pi_{0} =\{q_{1}\}$ and the set of secret states is $\Pi_{s} =\{q_{2}\}$. 
The corresponding OTS is shown in Fig.~\ref{fig: case I}. 
The reduced OTS (ROTS) is obtained by removing the states and transitions marked in red. 
Since there exists a path in the ROTS such that the state $(q_{2},\{q_{2},q_{4}\}, \{1\})$ is reachable, it follows from Theorem~\ref{theorem: verify CSEO} and Corollary~\ref{corollary: ROTS} that the transition system $T$ in this case is CSEO.

\emph{Case II:} Suppose that the set of initial states for the agent is $\Pi_{0} =\{q_{2}\}$ and the set of secret states is $\Pi_{s} =\{q_{4}\}$. 
The corresponding OTS is shown in Fig.~\ref{fig: case II}. 
Similarly, the ROTS is obtained by removing the states and transitions marked in red. 
In this case, any path in the OTS that visits the state $(q_{4},\{q_{2},q_{4}\}, \{1\})$ must also visit the state $(q_{4},\{q_{4}\}, \{2\})$. 
Moreover, the state $(q_{4},\{q_{2},q_{4}\}, \{1\})$ is not reachable in the ROTS. 
Therefore, according to Theorem~\ref{theorem: verify CSEO} and Corollary~\ref{corollary: ROTS}, the transition system $T$ in this case is not CSEO.

\begin{figure}[htbp]
	\centering
	\includegraphics[width=0.48\textwidth]{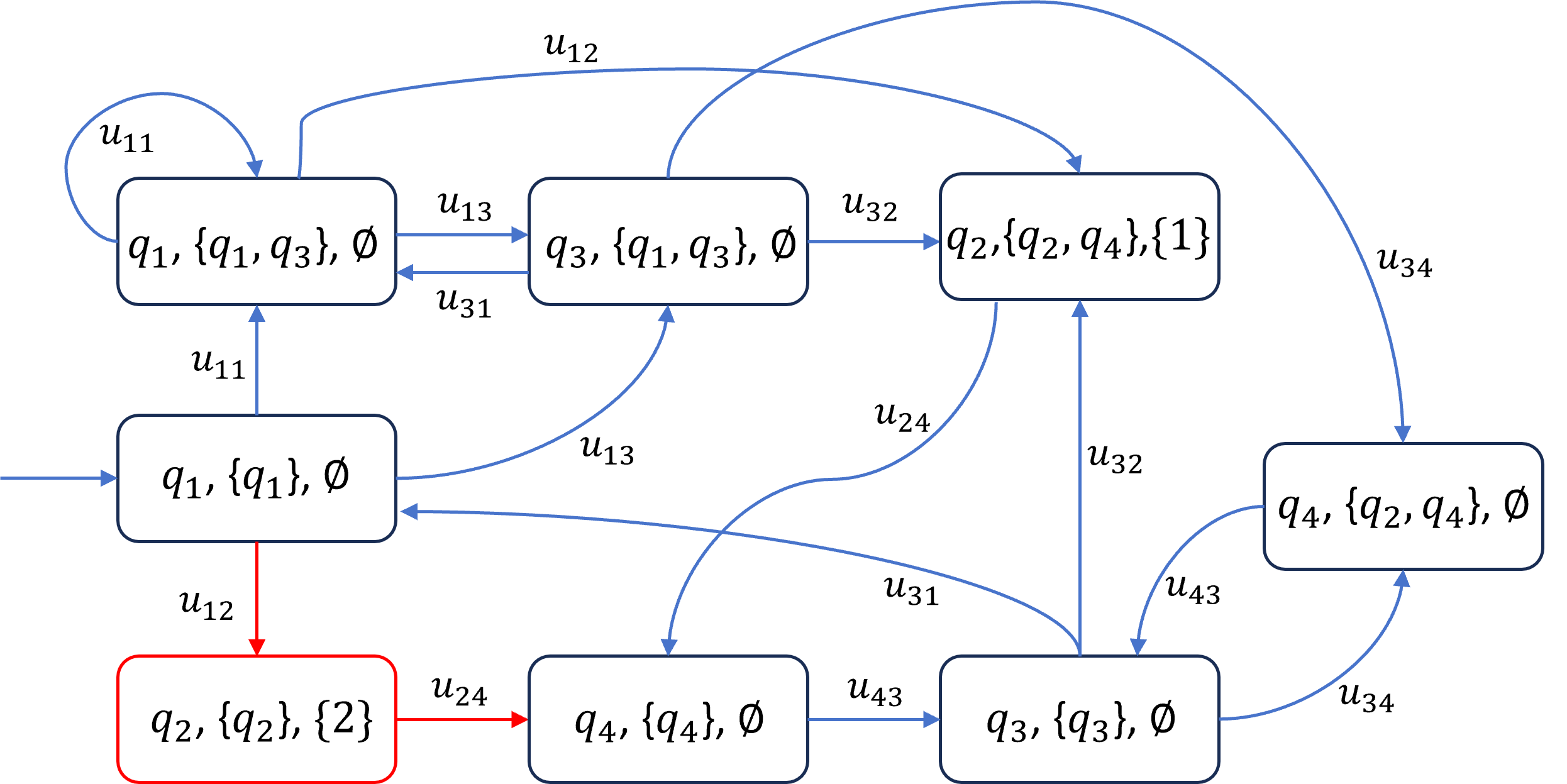}
    \vspace{-4pt}
	\caption{The OTS in Case I. The states and transitions marked in red are removed to obtain the corresponding ROTS.}
    \vspace{-4pt}
	\label{fig: case I}
\end{figure}

\begin{figure}[htbp]
	\centering
	\includegraphics[width=0.46\textwidth]{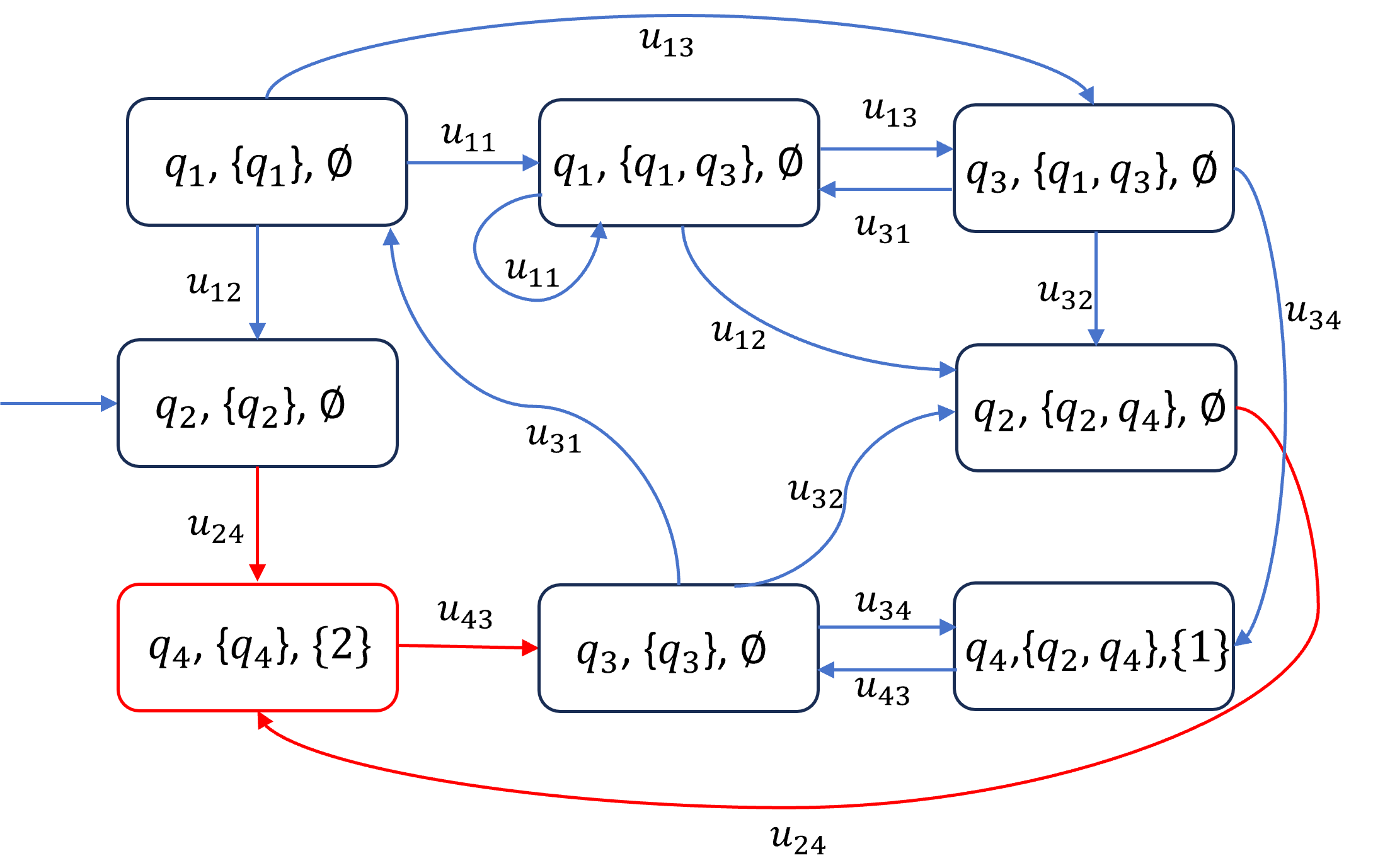}
    \vspace{-4pt}
	\caption{The OTS in Case II. The states and transitions marked in red are removed to obtain the corresponding ROTS.}
    \vspace{-4pt}
	\label{fig: case II}
\end{figure}
\vspace{-4pt}
\section{Conclusion} \label{Conclusion}

In this paper, we study state-observation-based opacity in transition systems observed by a passive intruder. Motivated by the limitations of this notion, we introduce a relaxed system property called \emph{existential opacity}. We further investigate the relationship between existential opacity and state-observation-based opacity, as well as the connection between existential opacity and the feasibility of opacity-preserving problems. In addition, we define a specific type of existential opacity, termed \emph{CSEO}, and propose a method to verify whether a system satisfies this property. 
For future work, we plan to explore additional types of existential opacity and develop corresponding verification approaches, as well as study their relation to specific opacity-preserving problems.

\bibliographystyle{ieeetr}
\bibliography{reference}

\end{document}